\documentclass[11pt,en,authoryear]{elegantpaper}

\theoremstyle{plain}
\newtheorem{assumption}{Assumption}
\newtheorem{property}{Property}

\newcommand{\EE}{\mathbb{E}}
\newcommand{\bds}{\boldsymbol}
\newcommand{\ind}{\perp\!\!\!\perp}
\newcommand{\ol}{\overline} 

\newcommand{\ul}{\underline}

\newcommand{\PRC}{\operatorname{PRC}}

\newcommand{\Cc}{\mathcal{C}}
\newcommand{\Kc}{\mathcal{K}}
\newcommand{\Jc}{\mathcal{J}}
\newcommand{\Rc}{\mathcal{R}}

\newenvironment{savableexample}[1]{%
    \def\currentexlabel{#1}%
    \begin{example}%
}{%
    \end{example}%
    \expandafter\xdef\csname saved@example@\currentexlabel\endcsname{\theexample}%
}

\newenvironment{continuedexample}[1]{%
    \edef\resumenum{\csname saved@example@#1\endcsname}%
    \addtocounter{example}{-1}%
    \expandafter\edef\csname theexample\endcsname{%
        \resumenum\ (continued)%
    }%
    \begin{example}%
}{%
    \end{example}%
    \renewcommand{\theexample}{\arabic{example}}%
}

\title{Probability of Root Cause: A Counterfactual Definition and Its Identification}
\author{Zitong Lu$^{1,2}$ \and Zhi Geng$^{3}$ \and Wei Li$^4$ \and Min Xie$^{1,*}$}
\institute{
	$^{1}$ Department of Systems Engineering, City University of Hong Kong, Hong Kong SAR \\ 
	$^{2}$ Department of Statistics and Data Science, The Chinese University of Hong Kong, Hong Kong SAR \\ 
	$^{3}$ School of Mathematics and Statistics, Beijing Technology and Business University, Beijing, China \\ 
	$^{4}$ Center for Applied Statistics and School of Statistics, Renmin University of China, Beijing, China 
}

\date{}

\begin{document}

\maketitle

\begin{abstract} 
Attributing an observed outcome to its root cause is a central task in domains ranging from medical diagnosis to engineering fault diagnosis. 
Existing approaches either equate the root cause with a root node of the causal graph, as in causal-discovery-based root cause analysis, or target causes more broadly and thereby favour proximate ones, as with the probability of causation and posterior causal effects. 
We argue that this issue stems from the absence of a formal definition of a root cause, which has led to methods designed for other purposes being applied to root cause attribution by default. 
We address this by giving a formal, individual-level definition of a root cause within the potential outcomes framework, based on the notion of an individual cause and a counterfactual root condition motivated by mediation analysis. 
Building on this definition, we propose the probability of root cause (PRC), which quantifies how probable it is that a candidate variable set is the root cause of a given outcome, conditional on observed evidence. 
Under standard assumptions, we establish the identifiability of the PRC and derive an explicit identification formula. 
Two numerical examples illustrate the approach. 
\keywords{Causal attribution; Causes of effects; Counterfactual; Probability of root cause; Root cause analysis}
\end{abstract}

\section{Introduction}

Attributing an observed outcome to its root cause is a central task across diagnostic settings, from medicine to engineering. 
In such settings, identifying a cause of an observed symptom is rarely enough: interventions on the proximate cause often provide only temporary relief, since the symptom recurs whenever the upstream root cause remains unaddressed. 
This distinction between a cause and a root cause is therefore not merely conceptual but practically consequential: isolating and intervening on the root cause, rather than its downstream manifestations, is essential for achieving durable resolution. 
Despite this practical importance, however, the statistical foundations of root cause attribution remain underdeveloped. 
Existing methods either operate at the population level, where individual-level attribution questions cannot be properly addressed, or fall under the broader umbrella of causal attribution, where the specific notion of root causes is not formally distinguished from causes. 
Three lines of work have engaged with this problem, each from a different angle. We review them in turn and explain why none directly solves root cause attribution. 

The first is root cause analysis (RCA), which directly targets root cause attribution and has gained increasing attention across diverse domains such as the internet of things, industry, and healthcare \citep{sole2017Survey, gao2015Survey,peerally2017Problem}. 
Quantitative approaches, including those based on Bayesian networks, fault tree analysis, and deep learning models \citep{xu2022Datadriven,cai2017Bayesian,lei2020Applications,shen2021Fault}, have been developed, yet achieving accuracy and explicability in root cause attribution algorithms remains a persistent challenge. 
Despite significant research efforts, most RCA methods operate at the association level, impeding the causal conclusions and reliable explanations at intervention and counterfactual levels, as formalized by \citet{pearl2018Book}. 
A notable strand of work seeks to incorporate causal interpretation through causal discovery \citep{li2025Root, lin2024Root, ikram2022Root, crane2024Root}. 
However, these methods typically equate a root cause with a root node of the causal graph, conflating two distinct notions. 
Moreover, causal discovery concerns constructing causal graphs from observational data, which operate at the population level and encode distributional causal relationships among variables. 
By contrast, attribution problems, including root cause analysis, are inherently individual-level issues, as emphasized by \citet{dawid2022Effects}.
For example, while diet may be identified as a root node for hypertension at the population level, a particular patient's hypertension may have an entirely different origin. 
Graph-based methods are, in general, not suitable for individual-level attribution problems.

The second is the causal attribution literature, which provides rigorous causal answers but to a related, broader question: what is the cause of an observed outcome? 
\cite{pearl1999Probabilities} proposed three indices to measure how probable an event is to be the cause of another. 
\cite{dawid2014Fitting} defined the probability of causation and emphasized the relationship between the causes of effects and the effects of causes. 
For scenarios involving multiple events influencing one another, \cite{lu2023Evaluating} introduced posterior causal effects to evaluate the likelihood that a specific event caused an occured outcome. 
Furthermore, \citet{li2024Retrospective} extended this framework to accommodate multiple outcome variables, and \citet{zhang2025Identifying} developed an analogous theory for ordinal outcomes.  
These methods provide principled answers to causes of effects, yet they share an inherent tendency to favour proximate causes over upstream ones, which will be discussed in Section~\ref{sec:notation}. 
The reason is structural: a cause and a root cause are different objects, and a measure designed to identify a cause cannot, in general, single out the root cause. 
 
Another related body of work is mediation analysis, which decomposes the total effect of a treatment on an outcome into components transmitted through specified pathways \citep{robins1992Identifiability, pearl2001Direct}.
An extensive literature has developed around natural and controlled direct and indirect effects, nonlinear models, unmeasured confounding, and multiple mediators \citep{imai2010General, shuai2026Mediation, vanderweele2014Mediation, zhou2022Semiparametrica}. 
Recent literature has begun to apply mediation analysis concepts to attribution analysis. 
\citet{rubinstein2025Mediated, kawakami2025Mediation} decompose the probability of causation into direct and indirect components to assess whether an observed adverse outcome was triggered through a mediating pathway or a direct pathway. 
These contributions refine attribution by distinguishing pathways. 
Mediation analysis decomposes effects along known pathways, treating pathways as the objects of inferential interest. 
Root cause attribution, by contrast, requires both identifying the relevant pathway and localizing its upstream origin. 
Adding pathway-level resolution to a method designed for causes of effects, therefore, produces a more refined answer to the same question, but not an answer to root cause attribution.


Despite this progress, a conceptual gap with practical consequences persists across both the RCA and causal attribution literatures. 
Although existing methods touch on the notion of root cause to varying degrees, none provide a formal definition of what a root cause is at the individual level. 
Without such a definition, the scope of applicability of those methods remains inherently ambiguous.

In this paper, we address this gap from first principles. 
Section \ref{sec:notation} introduces notation and preliminaries. 
In Section~\ref{sec:RC}, we formalize causality at the individual level via the notion of an {individual cause} and combine it with a {root condition} inspired by mediation analysis to give a formal definition of a root cause at the individual level. 
Building on this definition, we propose the {probability of root cause} (PRC), a counterfactual measure quantifying how probable it is that a candidate variable set is 
the root cause of an outcome, given the observed evidence. 
In Section \ref{sec:identification}, we establish the identifiability of the PRC, derive an explicit identification formula under standard assumptions, and characterize its relationship with posterior causal effects. 
Section \ref{sec:numerical-examples} illustrates the framework with two numerical examples drawn from \cite{tan2016Introduction} and \cite{yang2021Reliability}. 
We conclude with a discussion in Section \ref{sec:discussion}.

\section{Notation and Preliminaries}\label{sec:notation}
Let $X = (X_1,X_2,\dots,X_p)$ denote binary cause variables with $X_{i}=1$ representing the occurrence of cause event $X_i$, for example, illness, component error, and so on, and $X_i=0$ otherwise. 
Let $Y$ denote some binary outcome variable of interest with $Y=1$ representing the occurrence of outcome $Y$, for example, symptom, system failure, and so on, and $Y=0$ otherwise. 
And $Y$ is a descendant of causes $\bds X$. 
For convenience, let $\ol{X}_i = (X_1,\dots, X_{i-1})$ and $\ul{X}_i = (X_{i+1},\dots,X_{p})$. 
Without loss of generality, we assume that variables $X$ are in a topological order, which means for $1\leq i\leq p$, we have that $X_i$ is not a cause of $\overline{X}_i$, as shown in Fig.~\ref{fig:Xi}. 
We emphasize that our framework requires only this ordering, not the full causal graph; in particular, the conditional dependence structure within $X_{i}$'s ancestors need not be specified.
We use $\Omega$ to denote the sample space of individuals in the population and use $\omega$ for a particular sample point. 
For any variables (or variable sets) $W$ and $V$, let $W_{V=v}(\omega)$ denote the potential outcome of $W$, which means the value of $W$ for individual $\omega$ if $V$ is intervened to $v$. 
For simplification, we omit $\omega$ when there is no confusion. 
For example, if $W=(W_1,\dots,W_{n_{W}})$, we have $W_{v}=[(W_{1})_{v},\dots, (W_{n_{W}})_{v}]$.

\begin{figure}[!ht]
	\centering
	\includegraphics[width=0.5\linewidth]{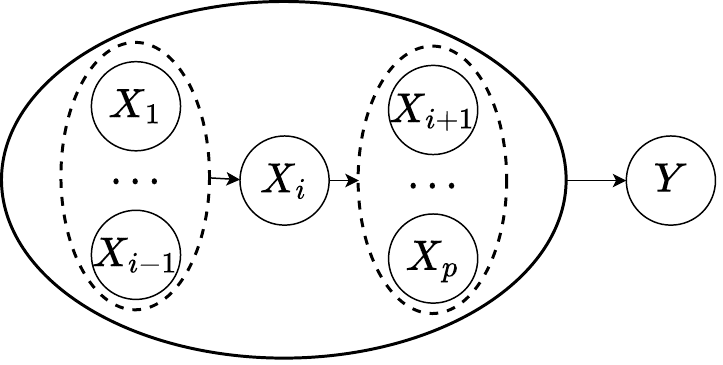}
	\caption[]{A general causal diagram for our model. For $i=1,\dots,p$, each $X_i$ has a similar structure to the figure. }
	\label{fig:Xi}
\end{figure}

In this paper, we suppose the consistency assumption holds, i.e., for any variable set $W$ and $V$, we have $W_{v}=W$ if $V=v$.  We also make a composition assumption, i.e., for any variable sets $W$, $V$, and $U$, we have $W_{uv}=W_{u}$ if $V_{u}=v$. 
These two assumptions connect the potential outcomes of a complex intervention to those of a simple intervention or even the observation. They are widely accepted as axioms of causal inference~\citep{pearl2009Causality}. 

When $p=1$, \cite{pearl1999Probabilities} defined probability of necessity
\begin{align*}
	PN = \Pr(Y_{X=0}=0\mid X=1,Y=1), 
\end{align*}
to measure how probable $X$ is a necessary cause of the effect $Y$. 
\cite{dawid2014Fitting} proposed the probability of causation as 
\begin{align*}
	PC = \Pr(Y_{X=0}=0\mid Y_{X=1}=1), 
\end{align*}
which could be equivalent to PN when $X$ is independent with $Y_{x}$, denoted as $X\ind Y_{x}$, according to \cite{dawid2015Causes}. 

For the situation with multiple $X$, which means $p>1$, \cite{lu2023Evaluating} defined the posterior total/direct causal effect of $X_k$ on $Y$ given $E=e$ as follows: 
\begin{align*}
	PostTCE(X_k\Rightarrow Y \mid E=e) &= \EE(Y_{X_k=1}-Y_{X_k=0}\mid E=e), \\
	PostDCE(X_k\Rightarrow Y_{\ul X_{k} = \ul x^{*}_{k}} \mid E=e) &= \EE(Y_{\ul X_{k} = \ul x^{*}_{k}}-Y_{X_k=0,\ul X_{k} = \ul x^{*}_{k}}\mid E=e), 
\end{align*}
to assess how likely $X_k$ is a cause of the effect $Y$ with observed evidence $E=e$. If $X_{k}=1,Y=1$ is included in $E=e$, then 
\begin{align*}
	PostTCE\left(X_k\Rightarrow Y \mid E=e\right) 
	= \Pr\left(Y_{X_k=0}=0\mid E=e \right), 
\end{align*}
which can be regarded as a generalization of PN. 

Despite the concept of `total causal effect' being intended to be distinguished from the `direct causal effect',  it ultimately falls short of reaching the realm of root cause due to the following property: 

\begin{property}\label{ppt:TCE}
	Under the Assumption~\ref{ass:no-conf} and \ref{ass:cross-world} of no confounding and cross-world independence, if $X_i$ blocks all the pathway from $X_k$ to $Y$ in the causal networks, we have
	\begin{align*}
		\bigl|PostTCE(X_k\Rightarrow Y\mid E=e)\bigr| &\leq \bigl|PostTCE(X_i\Rightarrow Y\mid E=e)\bigr| 
		. 
	\end{align*} 
	Further, if the Assumption~\ref{ass:mono} of monotonicity also holds, we have 
	\begin{align*} 
		PostTCE(X_k\Rightarrow Y\mid E=e) &\leq PostTCE(X_i\Rightarrow Y\mid E=e)
		. 
	\end{align*}
\end{property}

Property~\ref{ppt:TCE} illustrates that PostTCE is constrained by the inherent characteristics of the causal path and thus not a suitable measure for root cause attribution. 
See the supplementary materials \ref{proof:ppt_TCE} for proof. 
While understanding the total causal effect can provide valuable insights into the overall influence of a cause on an outcome, it may not provide specific information to identify the root cause. 
`Total causal effect' and `root cause' play different roles in causal mechanisms.

\section{Defining Root Cause at the Individual Level: Concepts and Probability} \label{sec:RC}

\subsection{Conceptual foundations: cause, root node, and root cause}


To formalize the concept of a root cause, we must first distinguish it from two related but distinct concepts: a `cause' and a `root node'. The failure to separate these ideas is a primary source of confusion in existing attribution methods. We use the simplified causal graph in Figure~\ref{fig:ex-simple} as a running example.

First, a `root cause' must, by definition, be a `cause'. This principle of causality requires that an event cannot be a cause if there is no causal pathway from it to the outcome \citep{richens2020Improving}. For instance, in Figure~\ref{fig:ex-simple}, `Chest Pain' is a consequence of `Heart Disease', not a cause of `High Blood Pressure'. While it may be correlated with hypertension, it cannot be its cause. This principle immediately rules out purely association-based attribution methods.

Second, a `root cause' should not be conflated with a `root node' in a causal graph. A root node (e.g., `Diet') represents a variable that is not caused by any other variable in the graphical model. However, attribution is fundamentally an individual-level query \citep{dawid2022What}, whereas a causal graph represents population-level tendencies. 
Consider a patient with high blood pressure: even if `Diet' is a root node in the causal graph, it cannot be a cause for a patient who maintains a healthy diet. Moreover, for a patient with an unhealthy diet, individual-level causal mechanisms may dictate that their diet does not actually affect their blood pressure, and thus `Diet' would not be the root cause. 

Finally, and most critically, we must distinguish a `root cause' from other causes. Consider a patient with hypertension, where an unhealthy `Diet' has led to `Heart Disease', which in turn causes `High Blood Pressure'. Both `Diet' and `Heart Disease' are causes. However, existing causal attribution methods, such as those based on PostTCE, will invariably assign more weight to the more proximate cause (`Heart Disease'), as guaranteed by Property~\ref{ppt:TCE}. While medically correct, this may not be what is desired from a root cause analysis. The goal is often to identify the upstream factor that, if addressed, would resolve the issue most effectively. Intervening on the heart disease may be a temporary fix if the underlying dietary issue persists.

\begin{figure}[!ht]
	\centering
	\includegraphics[width=0.3\linewidth]{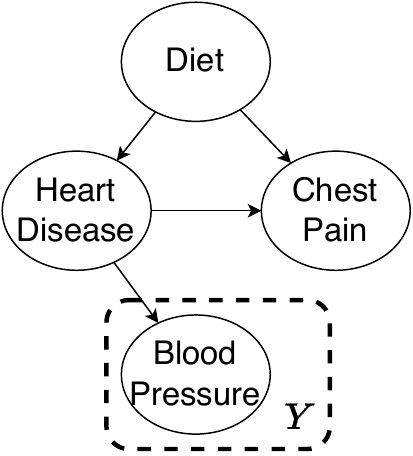}
	\caption[]{A simplified causal network representing hypertension and its risk factors. A more detailed diagram will be given in Section~\ref{sec:numerical-examples}. }
	\label{fig:ex-simple}
\end{figure}

These distinctions motivate our framework. A proper definition of a root cause must be based on individual-level causality and must be able to look beyond proximate causes to identify the origin of the causal chain, without being restricted to root nodes.

\subsection{Cause at the individual level}

While a causal graph may suggest that `Diet' is a cause of `High Blood Pressure' at the population level, this need not hold for a specific individual. 
To build our framework, we first require a precise, individual-level definition of a cause. 
We therefore ground our framework in the potential outcomes notation, aiming to capture the causal mechanisms that govern the relationship between variables for a specific individual. 

\begin{definition}[Individual Cause]\label{def:IC}
    A set of variables $V=(V_{1},\dots,V_{s})$ constitutes an individual cause (IC) of a variable $W$ for a specific individual $\omega$ if there exist two different interventions $v$ and $v'$ such that $W_{V=v}(\omega)\neq W_{V=v'}(\omega)$. 
\end{definition}
This definition formalizes the intuitive notion that $V$ is a cause of $W$ if intervening on $V$ can alter the outcome $W$ for that individual.
Unlike more specific notions such as sufficient or necessary causation, which focus on activating or preventing the occurred event, respectively, depending on the realization, IC concerns any alteration of the outcome event, regardless of the observed realization. 
The concept of IC is closely related to the individual treatment effect (ITE) \citep{neyman1923Application,rubin1974Estimating}.
For a binary treatment $V$, being an IC of $W$ for individual $\omega$ is equivalent to
a non-zero ITE, $\tau(\omega) = W_{V=1}(\omega) - W_{V=0}(\omega) \neq 0$. 
More broadly, \citet{lu2023Evaluating} show that various causal attribution concepts can be unified through individual treatment effects. 
From this perspective, sufficient and necessary causation correspond to a non-zero ITE conditional on $(V(\omega), W(\omega)) = (1,1)$ and $(0,0)$, respectively, and are thus special cases of IC subject to additional restrictions on the observed realization.  
The IC relation also satisfies two natural properties.
First, if $V$ is an IC of $W$ for $\omega$, then any superset $V^{\sup} \supseteq V$ is also an IC of $W$ for $\omega$. 
Second, ICs need not be unique; for instance, if $Y(\omega) = X_1(\omega) \land X_2(\omega)$ with $X_{1}\ind X_{2}$,  then both $X_1$ and $X_2$ are ICs of $Y$ for $\omega$ with $Y(\omega) = 1$. 

The distinction between individual and population-level causality is critical for attribution problems. 
Consider a simple causal model $X \to Y$. 
Using principal stratification \citep{frangakis2002Principal}, the population can be divided into four types based on their potential outcomes: compliers, defiers, always-takers, and never-takers. 
At the population level, $X$ can be regarded as a cause of $Y$. 
But at the individual level, it is not suitable to attribute $Y$ to $X$ for always-takers and never-takers, because intervening on $X$ has no effect on $Y$. 
Our definition correctly captures such individual-level nuance, which is essential for attribution.

\begin{table}[!ht]
	\caption{We list all kinds of individuals with different $(L_{D},L_{Hd},L_{Bp})$ values. If $D\to Hd\to Bp$ forms a causal chain as Figure~\ref{fig:ex-simple}, the corresponding root causes of $Bp$ under observation $(D,Hd,Bp)=(1,1,1)$ are listed in the second-to-last column. \label{tab:ex-3chain}}
	\centering
	\begin{tabular}{cccccc}
		\hline
		Individual Type & $L_{D}$ & $L_{Hd}$ & $L_{Bp}$     & Root Cause of $Bp$ & Individual Causal Diagram\\ \hline
		1     & 0     & 0       & 0       & No Root Cause, or $Bp$ itself                & $D,\quad Hd,\quad Bp$\\
		2     & 0     & 1       & 0       & No Root Cause, or $Bp$ itself                   & $D\to Hd,\quad Bp$\\
		3     & 0     & 0       & 1       & $Hd$               & $D,\quad Hd\to Bp$\\
		4     & 0     & 1       & 1       & $D$               & $D\to Hd\to Bp$\\ \hline
		\end{tabular}
\end{table}

We use the example in Figure~\ref{fig:ex-simple} to illustrate how the concept of IC can help us understand the problems of identifying root causes. 
Let $(D,Hd,Bp)$ form a causal chain diagram as Figure~\ref{fig:ex-simple}, and represent whether an individual has an unhealthy diet, heart disease, and high blood pressure, respectively, with 1 indicating presence and 0 absence. 
Chest pain is excluded from the discussion, for it does not influence blood pressure. 
Given the evidence, $(D,Hd,Bp)=(1,1,1)$, one may naively attribute the root cause of $Bp=1$ to $D=1$ by the causal graph or attribute it to $Hd=1$ by the PostCEs method.  
However, both methods could be incorrect due to the causal mechanisms. 
For a variable $W$, let $L_{W}(\omega)$ denote whether $W$ has an IC for $\omega$; and denote $L_{W}(\omega)=1$ if and only if there exists a set $V$ constituting an IC of $W$ for $\omega$. 
We will omit $\omega$ when there is no confusion. 
In a DAG, if there are no parent variables of $W$, $Pa(W)=\emptyset$, then we set $L_{W}=0$. 
If $L_{Bp}=0$, then we can not change the value of $Bp$ by intervening on other variables. Thus, there is no IC or root cause of $Bp$ among $D$ and $Hd$, or we could say $Bp$ is the root cause of itself. 
If $(L_{D},L_{Hd},L_{Bp})=(0,0,1)$, then we can change the value of $Bp$ by intervening on $Hd$, but not by intervening on $D$; thus $Hd$ is an IC of $Bp$ and there is no IC of $Hd$, $Hd$ should be the root cause rather than $D$. 
If $(L_{D},L_{Hd},L_{Bp})=(0,1,1)$, then we can change the value of $Hd$ by intervening on $D$, and thereby affect the value of $Bp$. 
In this case, $D$ should be the root cause of $Bp=1$, while $Hd$ is a cause but not a root cause. 
From a perspective of graphs, the causal diagram is the sum of individual causal diagrams; we should do attribution based on the individual causal diagram rather than the population-level DAG. 
In practice, one patient with observation $(D,Hd,Bp)=(1,1,1)$ can belong to any kind listed in Table~\ref{tab:ex-3chain}, and any of $Bp$, $Hd$, and $D$ could be the root cause of high blood pressure due to the latent causal mechanisms. 

The above example may suggest that root causes can be identified recursively by repeatedly tracing the causes of causes. 
However, IC does not satisfy transitivity: a cause of a cause of an effect is not necessarily itself a cause of that effect. 
\begin{savableexample}{ex:or}\label{ex:or}
	Given the model with
	\begin{align*}
		Y(\omega) = X_1(\omega) \lor X_2(\omega), \qquad X_2(\omega) = 1-X_1(\omega).
	\end{align*} 
	Then $X_1$ is an IC of $X_2$, and $X_2$ is an IC of $Y$, but $X_1$ is not an IC of $Y$ (since $Y$ is always 1). Thus, $X_{1}$ is obviously not a root cause of $Y$. 
\end{savableexample}
This non-transitivity highlights that one cannot simply trace back a chain of ICs to find a root cause; a more sophisticated definition of root cause is required.

\subsection{Definition and probability of root cause}

The root cause, as its name suggests, should meet two conditions:
`cause' and `root'. 
For the cause condition, a root cause must be an IC, which is defined and discussed in previous sections. 
But for the root condition, there is no clear definition in the literature. 
From the perspective of a causal path, finding the `root' means finding the start of a causal path. As discussed in the no-transitivity, even for the same individual causal diagram, the underlying causal paths differ. 
This prevents us from simply identifying the root through the graph structure alone, even for individual causal diagrams; we need to further decompose the causal paths. 



Based on the two conditions of root cause, we propose our definition of the root cause. 
Our definition here does not depend on binary variables and can be easily extended to non-binary variables. 

\begin{definition}[Root Cause]\label{def:RC}
	A binary variable set $V$ is said to constitute a root cause (RC) of the binary scalar variable $W$ for a specific individual $\omega$ if and only if the following conditions hold:
	\begin{enumerate}[(a)]
		\item\label{itm:cause} $V$ is the IC of $W$ for $\omega$ (cause); 
		\item\label{itm:root} There is no indirect effect on $W$ through $V$, specifically, for any vector $U$ with $U\cap V=\emptyset$ and $U\cap W=\emptyset$ and its values $u\neq u'$, we have $W_{V_{u}}(\omega)=W_{V_{u'}}(\omega)$ (root).  
		%
		%
		%
	\end{enumerate}
\end{definition}
We regard root causes as a special case of causes: an RC must be an IC, which is the first condition. 
As $V$ is the IC of $W$ for $\omega$, there is a causal pathway from $V$ to $W$ for this individual. 
The second condition sets up a criterion for verifying whether $V$ is the origin of that causal pathway to $W$. 
This condition uses a nested counterfactual to check if any other variable $U$ acts as an ancestor of $W$ through $V$. 
It requires that even if an intervention on $U$ may alter the value of $V$, such a change in $V$ is not allowed to affect $W$ indirectly. 
If it did, $U$ (or one of its components) would be a more fundamental cause acting through the causal pathway $U\to V\to W$, and thus $V$ would not be the `root'. 
This condition is closely related to the concept of indirect effects~\citep{robins1992Identifiability,pearl2001Direct}
\begin{align*}
	IE = W_{uV_{u}}-W_{uV_{u'}}. 
\end{align*}
Accordingly, we can rewrite the root condition in the definition in the form of $W_{U,V_{U}}(\omega)-W_{U,V_{u'}}(\omega)=0$ for all $u'$, which indicates the absence of an indirect effect of $U$ on $W$ through the mediator $V$. 
When the root condition is violated, $U$ or its elements of $U$ can be incorporated into the causal pathway to $W$ as the ancestors of $V$ and $W$. 
It is worth noting that the RC need not be unique, which means there could be multiple RCs of an effect for a specific individual. 
Consider the illustrative example where $Y=X_{1}\land X_{2}$ and $X_{1}\ind X_{2}$: when $Y(\omega)=1$, both $X_{1}$, $X_{2}$, as well as the variable set $(X_{1},X_{2})$, can serve as RCs of $Y$ for $\omega$. 

When $(X_{1},\dots,X_{p},Y)$ constitute a causal chain diagram, we have the following corollary:
\begin{corollary}[RC for Causal Chain Diagram]
	For a causal chain DAG, $X_{1}\to X_{2}\to \dots \to X_{p}\to Y$, a variable $X_{k}$ is said to be the RC of $Y$ for $\omega$ if and only if the following conditions hold: 
	\begin{enumerate}[(a)]
		\item $L_{Y}(\omega)= L_{X_{p}}(\omega) =\dots=L_{X_{k+1}}(\omega)=1$, 
		\item $L_{X_{k}}(\omega)=0$, 
	\end{enumerate}
	where $L_{X_{k}}(\omega)$ is the indicator whether $X_{k}$ has an IC for $\omega$.
\end{corollary}
The definition of RC is consistent with the analysis of the causal chain case in Table~\ref{tab:ex-3chain}, showing that our definition is suitable and can capture the intuitive notion of root cause in a simple causal chain. 

The Definition~\ref{def:RC} of RC is based on the individual-level causal mechanisms, which are unobservable in practice. 
Therefore, we further define the probability of root cause to assess how probable a variable set is the root cause of the outcome variable, under the model with variables $(X_{1},\dots,X_{p},Y)$. 
For convenience, with a index set $\Kc=\{k_{1},\dots, k_{r}\}$ with $1\leq k_{1}<k_{2}< \dots < k_{r}\leq p$, we denote $X_{\Kc}=(X_{k_{1}},\dots, X_{k_{r}})$, $\ol \Kc = \{j\mid 1\leq j<k_{1}\}$, $\ul \Kc = \{j\mid k_{r}<j\leq p\}$, and $\hat \Kc= \{j \mid k_{1}\leq j\leq k_{r}\}\setminus \Kc  = \{j\mid k_{1}<j<k_{r},j\notin \Kc\}$. 
Let $E=e$ denote the observed evidence for the individual we want to analyse, where $E\subseteq \{X_{1},\dots,X_{p},Y\}$.

We first define two indicators for the cause and root conditions in Definition~\ref{def:RC}: 
\begin{align}
	\Cc_{X_{\Kc}}^{Y} &= \bigvee_{x_{\Kc}'\neq x_{\Kc}''} \{Y_{x_{\Kc}'}\neq Y_{x_{\Kc}''}\}, \label{eq:C-indicator}
	\\
	\Rc_{X_{\Kc}}^{Y} &= \bigwedge_{x_{\ol \Kc \cup \hat \Kc}'\neq x_{\ol \Kc\cup \hat \Kc}''} \bigl\{Y_{(X_{\Kc})_{x_{\ol \Kc\cup \hat \Kc}'}}= Y_{(X_{\Kc})_{ x_{\ol \Kc\cup \hat \Kc}''}}\bigr\}. \label{eq:R-indicator}
\end{align}
Here $\Cc_{X_{\Kc}}^{Y}=1$ indicates that $X_{\Kc}$ is the IC of $Y$, and $0$ otherwise. 
$\Rc_{X_{\Kc}}^{Y}=1$ indicates that there is no indirect effect through $X_{\Kc}$ on $Y$, i.e., $X_{\Kc}$ is the origin of the causal pathway to $Y$, and $0$ otherwise. And we denote $\Rc_{X_{\Kc}}^{Y}=1$ if $X_{\ol \Kc\cup \hat \Kc}=\emptyset$, which means if there is no ancestor of $X_{\Kc}$, then $X_{\Kc}$ satisfies the root condition by default. 
Based on these two indicators, we can define the probability of root cause as follows.

\begin{definition}[Probability of Root Cause]\label{def:PRC}
	Given the observed evidence $E=e$, for a nonempty set $\Kc\subseteq \{1,\dots,p\}$, we can define the probability that $X_{\Kc}$ constitute a RC of $Y$ as follows:
	\begin{equation}\label{eq:PRC}
		\PRC(X_{\Kc}\Rightarrow Y\mid E=e) = \Pr\bigl(\Cc_{X_{\Kc}}^{Y} = 1, \Rc_{X_{\Kc}}^{Y}=1\mid E=e\bigr), 
	\end{equation}
    where $\Cc_{X_{\Kc}}^{Y}$ and $\Rc_{X_{\Kc}}^{Y}$ are defined in equations \eqref{eq:C-indicator} and \eqref{eq:R-indicator}. 
\end{definition}

The observed evidence $E = e$ is used to condition the probability, enabling us to incorporate available information about a specific individual into our assessment of RC. 
This conditioning is particularly important because RC is defined in terms of individual potential outcomes, which are inherently unobservable in practice. The observed evidence thus provides partial information about these potential outcomes, thereby improving the precision of the causal assessments.
In the medical diagnosis context, let $X$ denote the diseases and $Y$ denote the symptom. 
The observed evidence $E = e$ captures the patient's observed symptom profile and disease status, which can be leveraged to enable personalized diagnosis. 
When $E=\emptyset$, no information is available, and the probability of root cause thus reflects population-level distribution. 
In contrast, when $\{E=e\}=\{X_{1}=1,Y=1\}$, we observe both the presence of disease $X_{1}$ and a positive symptom status, allowing the PRC to be updated to reflect this individual-specific evidence. 

As Table \ref{tab:ex-3chain} discussed, there is the case that $Y$ does not have a root cause, or we could say, $Y$ is the root cause of itself. 
Based on the two indicators \eqref{eq:C-indicator} and \eqref{eq:R-indicator}, replacing $X_{\Kc}$ with $Y$, we denote 
\begin{align}
	\Cc_{Y}^{Y} &= \{Y_{Y=1}\neq Y_{Y=0}\}=1, \label{eq:C-Y}
	\\
	\Rc_{Y}^{Y} &= \bigwedge_{x'\neq x''} \{Y_{Y_{x'}}=Y_{Y_{x''}}\} = \bigwedge_{x'\neq x''} \{Y_{x'}=Y_{x''}\}, \label{eq:R-Y}
\end{align} 
where $\Cc_{Y}^{Y}=1$ because $Y$ is always an IC of itself as intervening on $Y$ can change the value of $Y$; and $\Rc_{Y}^{Y}$ is the indicator if there is no cause of $Y$ other than $Y$ itself. 
Thus, we can define the probability that there is no RC of $Y$ as the probability that $Y$ itself is the RC of $Y$ as follows. 
\begin{definition}[Probability of No Root Cause]\label{def:PRC-Y}
	Given the observed evidence $E=e$, we define the probability that there is no RC of $Y$ as
	\begin{equation*}
		\PRC(\emptyset\Rightarrow Y\mid E=e) 
		= \PRC(Y\Rightarrow Y\mid E=e)
		= \Pr\bigl(\Cc_{Y}^{Y}=1,\Rc_{Y}^{Y}=1\mid E=e\bigr), 
	\end{equation*}
    where $\Cc_{Y}^{Y}$ and $\Rc_{Y}^{Y}$ are defined in equations \eqref{eq:C-Y} and \eqref{eq:R-Y}. 
\end{definition}
A simple corollary is that $\PRC(\emptyset\Rightarrow Y\mid E=e)$ equals the probability that $X=(X_{1},\dots,X_{p})$ is not the RC of $Y$: 
\begin{align}\label{eq:PRC-Y-X}
	\PRC(\emptyset \Rightarrow Y\mid E=e) = 1 - \PRC(X\Rightarrow Y\mid E=e).
\end{align}

We revisit Example~\ref{ex:or} to illustrate how the two indicators $\Cc_{X_{\Kc}}^{Y}$ and $\Rc_{X_{\Kc}}^{Y}$ work together to determine the PRC. 
\begin{continuedexample}{ex:or}
	The causal mechanisms are given by $Y=X_{1}\lor X_{2}$ and $X_{2}=1-X_{1}$. 
	\begin{itemize}
		\item For $X_{1}$: $\Cc_{X_{1}}^{Y}=\{Y_{X_{1}=1}\neq Y_{X_{1}=0}\}=0$ and $\Rc_{X_{1}}^{Y}=1$, so $\PRC(X_{1}\Rightarrow Y\mid E=e)=0$. Here $X_{1}$ satisfies the root condition but is not an IC of $Y$. 
		\item For $X_{2}$: $\Cc_{X_{2}}^{Y}=\{Y_{X_{2}=1}\neq Y_{X_{2}=0}\}=1$ and $\Rc_{X_{2}}^{Y}=0$, so $\PRC(X_{2}\Rightarrow Y\mid E=e)=0$. Here $X_{2}$ is an IC of $Y$ but does not satisfy the root condition. 
		\item For the set $X_{\{1,2\}}$: $Y_{00}\neq Y_{11}$ and $X_{\{1,2\}}$ has no ancestor, so $\Cc_{X_{\{1,2\}}}^{Y}=1$ and $\Rc_{X_{\{1,2\}}}^{Y}=1$, giving $\PRC(X_{\{1,2\}}\Rightarrow Y\mid E=e)=1$. The set $X_{\{1,2\}}$ is the root cause of $Y$. 
	\end{itemize} 
\end{continuedexample}

There are some properties of PRC: 
\begin{property}\label{ppt:PRC}\ 
	\begin{enumerate}[(a)]
		\item\label{enu:IC} 
		If $X_{\mathcal{J}}$ is an IC of $Y$, then any superset $X_{\Kc}$ containing $X_{\mathcal{J}}$ is also an IC of $Y$. 
		But if $X_{\mathcal{J}}$ is the RC of $Y$, a superset $X_{\Kc}$ need not be an RC of $Y$. 
		\item\label{enu:k} If $\Kc=\{k\}$, we have $\Rc_{X_{k}}^{Y}=1-L_{k}$,
			\begin{equation}\label{eq:PRC-k}
				\PRC(X_{k}\Rightarrow Y\mid E=e) 
				= \Pr\bigl(Y_{X_{k}=0}\neq Y_{X_{k}=1}, L_{k}=0\mid E=e\bigr), 
			\end{equation}
			where $L_{k}=0$ if and only if there are no variables that constitute an IC of $X_{k}$. 
		\item\label{enu:1-k} If $\Kc=\{1,\dots,k\}$, we have $\Rc_{X_{\Kc}}^{Y} = 1$ and 
			\begin{align*}
				\PRC(X_{\Kc}\Rightarrow Y\mid E=e) &\geq  \PRC(X_{\mathcal{J}}\Rightarrow Y\mid E=e),  
			\end{align*}
			for any $\mathcal{J}\subseteq \Kc$. 
		\item\label{enu:non-causation} If there is no causal pathway from $X_{i}$ to $Y$, then 
			\begin{align*}
				\PRC(X_{\Kc\setminus\{i\}}\Rightarrow Y\mid E=e) &= \PRC(X_{\Kc}\Rightarrow Y\mid E=e).
			\end{align*}
	\end{enumerate}
\end{property} 

Proof is given in the supplementary material~\ref{proof:ppt_PRC}.
Property~\ref{ppt:PRC}\ref{enu:IC} states that the IC condition is monotone under set inclusion, whereas the RC condition is not. Intuitively, adding a variable $X_{k}$ to a root-cause set 
$X_{\mathcal{J}}$ may fail to preserve the root condition, because $X_{k}$ may itself possess an IC outside $X_{\mathcal{J}}$ that has an indirect effect on $Y$. 
Property~\ref{ppt:PRC}\ref{enu:k} specializes the definition to a single variable. 
Equation~\eqref{eq:PRC-k} reveals that the root condition reduces to the absence of any IC of $X_{k}$, thus the PRC equals the probability that $X_{k}$ is an IC of $Y$ while itself having no IC. 
This property aligns with the intuition on root cause and the discussion in Table~\ref{tab:ex-3chain}. 
Property~\ref{ppt:PRC}\ref{enu:1-k} shows that when $X_{\Kc}$ has no ancestors, the root condition holds automatically, and the PRC is monotonically non-decreasing in the candidate set.
Property~\ref{ppt:PRC}\ref{enu:non-causation} says that variables with no causal pathway to $Y$ are irrelevant: removing such a variable from $X_{\Kc}$ leaves the PRC unchanged.





\section{Identifiability of Probability of Root Cause} \label{sec:identification}

For any variable vector $W=(W_1,\dots, W_s)$ and its value $w$ and $w^*$, let \(w^*\preceq w\) denote that \(w_i^*\leq w_i\) for all \(1\leq i\leq s\). 
To identify the probability of the root cause, we need some commonly-used assumptions: monotonicity, no confounding, and cross-world independence assumptions \citep{pearl2009Causality,lu2023Evaluating,zhao2023Conditional}. 

\begin{assumption}[Monotonicity]\label{ass:mono}
	Let \(k=\min \Kc\), 
	\begin{enumerate}[(a)]
		\item \(X\) satisfies the monotonicity in the sense that $(X_i)_{\ol{X}_i=\ol{x}_i^{(0)}} \leq (X_i)_{\ol{X}_i=\ol{x}_i^{(1)}}$ for any \(i\geq k\) whenever \(\ol{x}_i^{(0)} \preceq \ol{x}_i^{(1)}\); 
		\item \(Y\) satisfies the monotonicity in the sense that $(Y)_{X=x^{(0)}} \leq (Y)_{X=x^{(1)}}$ whenever $x^{(0)} \preceq x^{(1)}$.  
	\end{enumerate} 
\end{assumption}
\begin{assumption}[No Confounding]\label{ass:no-conf}\ 
	\begin{enumerate}[(a)]
		\item\label{enu:no-conf-1} $(X_i)_{\ol{x}_i}\ind \ol X_{i}$, for $1\leq i\leq p$;
		\item\label{enu:no-conf-2} $Y_{x} \ind X$. 
	\end{enumerate}
\end{assumption}
\begin{assumption}[Cross World Independence]\label{ass:cross-world}
	Potential outcomes of each variables, $Y_{X=x}$ and $(X_i)_{\ol{X}_i = \ol{x}_i}$ for $1\leq i\leq p$ are independent. 
\end{assumption}
The monotonicity assumption means that the causal mechanisms of each variable are never preventive. 
For example, if we cure some diseases (or repair some parts of a machine), the symptoms of the patient (or the operational states of that machine) will not get worse. 
Assumption \ref{ass:no-conf} is also known as `exogeneity' in the econometrics literature. 
It means that the potential outcomes of each variable are independent of its ancestors. 
It is worth noting that Assumption \ref{ass:cross-world} could be mild under the structural equation model (SEM): 
If $(X,Y)$ satisfies a nonparameteric SEM $X_i = f_{i}(\ol X_i;\epsilon_{i})$ with $\epsilon_{i}\ind \ol X_i$ and $Y= f_{Y}(X;\epsilon_{Y})$ with $\epsilon_{Y}\ind X$, Assumption \ref{ass:no-conf}\ref{enu:no-conf-1} and \ref{ass:no-conf}\ref{enu:no-conf-2} imply $(\epsilon_1,\dots,\epsilon_p,\epsilon_Y)$ are independent, thus Assumption \ref{ass:cross-world} holds.


Under Assumptions \ref{ass:mono}, \ref{ass:no-conf}, and \ref{ass:cross-world} of monotonicity, no confounding, and cross-world independence, we show the identifiability of the probability of root cause and its identification equation. 

\begin{theorem}\label{thm:PRC}
	Given the evidence \(E=e\) where $E\subseteq (X,Y)$, we have 
	\begin{equation*}\begin{aligned}
		& \PRC(X_{\Kc}\Rightarrow Y \mid E=e) 
        \\
        =& \sum_{(x,y)} \PRC(X_{\Kc}\Rightarrow Y \mid X=x,Y=y) \times \Pr(X=x,Y=y\mid E=e)
		. 
	\end{aligned}\end{equation*}
	If Assumptions \ref{ass:mono}, \ref{ass:no-conf} and \ref{ass:cross-world} hold, \(\PRC(X_\Kc\Rightarrow Y\mid X=x, Y=y)\) can be identifiable. The identification formula is 
	\begin{equation}\label{eq:PRC-id}\begin{aligned}
		&\quad \PRC(X_{\Kc}\Rightarrow Y\mid X=x,Y=1) 
		\\
		&= \sum_{\substack{x^*\preceq x' \preceq x: \\ x_{\ol \Kc}^{*}=x_{\ol \Kc}'=x_{\ol \Kc},\ x_{\Kc}^{*}=0}} \frac{\Pr(Y=1 \mid x') - \Pr(Y=1 \mid x^{*})}{\Pr(Y=1\mid X=x)} 
        \\&\qquad\quad \times \prod_{i\in \Kc} \Bigl\{ (1-x_{i}') + x_{i} (2x_{i}'-1) \cdot \frac{\Pr(X_{i}=1 \mid \ol X_{i}=\ol x_{i}'')}{\Pr(X_{i}=1\mid \ol X_{i}=\ol x_{i})} \Bigr\}
		\\&\qquad\quad \times \prod_{i\in \hat \Kc\cup \ul \Kc} \Bigl\{ (1-x_{i}')(1-x_{i}^*) + x_{i}(2x_{i}'-1)(1-x_{i}^{*}) \frac{\Pr\bigl(X_{i}=1\mid \ol X_{i}=\ol x_{i}'\bigr)}{\Pr(X_{i}=1\mid \ol X_{i}=\ol x_{i})} 
		\\&\qquad\qquad\qquad\qquad + x_{i} x_{i}'(2x_{i}^{*}-1) \frac{\Pr\bigl(X_{i}=1\mid \ol X_{i}=\ol x_{i}^*\bigr)}{\Pr(X_{i}=1\mid \ol X_{i}=\ol x_{i})} \Bigr\}
		, 
	\end{aligned}\end{equation}
	where $x_{\Kc}''=x_{\Kc}'$, $x_{\ol \Kc\cup \hat \Kc}''=0$. And 
	\begin{equation*}\begin{aligned}
		\PRC(X_{\Kc}\Rightarrow Y\mid X=x,Y=0) 
		= \PRC(\tilde X_{\Kc}\Rightarrow \tilde Y\mid \tilde X=1-x,\tilde Y=1) 
		, 
	\end{aligned}\end{equation*}
	where $\Pr(\tilde X,\tilde Y) = \Pr(X=1-x,Y=1-y)$. 
\end{theorem}
Theorem~\ref{thm:PRC} shows that $\PRC(X_{\Kc}\Rightarrow Y \mid E=e)$ can be expressed as a weighted sum of $\PRC(X_{\Kc}\Rightarrow Y \mid X=x,Y=y)$, which is identified under Assumptions \ref{ass:mono} and \ref{ass:no-conf}. 
Proof can be found in the supplementary material~\ref{proof:thm_PRC}. 
As a corollary, the probability of no root cause is also identifiable based on the equation \eqref{eq:PRC-Y-X}. 
\begin{corollary}
	The probability of no root cause is identifiable under Assumptions \ref{ass:mono}, \ref{ass:no-conf}, and \ref{ass:cross-world}. 
	\begin{align*}
		\PRC(\emptyset \Rightarrow Y\mid E=e) = 1- \PRC(X\Rightarrow Y\mid E=e), 
	\end{align*}
	where the identification formula of $\PRC(X\Rightarrow Y\mid E=e)$ is given by Theorem~\ref{thm:PRC}. 
\end{corollary}

The identification formula \eqref{eq:PRC-id} of PRC is similar but more complex than PostTCE~\citep{lu2023Evaluating}, and we have the following properties about their relationships.

\begin{property}\label{ppt:PRC-TCE}\ 
	\begin{enumerate}[(a)]
		\item When $\Kc=\{k\}$, if assumptions \ref{ass:mono}, \ref{ass:no-conf} and \ref{ass:cross-world} hold, we have the following equation: 
			\begin{equation*}\begin{aligned}
			    &\PRC(X_{k}\Rightarrow Y \mid X=x,Y=1) 
				\\
                =& x_{k}\frac{\Pr(X_{k}=1\mid \ol X_{k}=0)}{\Pr(X_{k}=1\mid \ol X_{k}=\ol x_{k})} \cdot \operatorname{PostTCE}(X_{k}\Rightarrow Y\mid X=x,Y=1)
				. 
			\end{aligned}
			\end{equation*}
		\item When $\Kc=\{1,\dots,k\}$, if assumptions \ref{ass:mono}, \ref{ass:no-conf} and \ref{ass:cross-world} hold, we have the following equation: 
			\begin{align*}
				\PRC(X_{\Kc}\Rightarrow Y\mid E=e) = \operatorname{PostTCE}(X_{\Kc}\Rightarrow Y \mid E=e) 
				. 
			\end{align*}
	\end{enumerate}
\end{property}
In the special example $p=2$ with $\{E=e\}=\{(X_{1},X_{2},Y)=(1,1,1)\}$, the two measures are 
related as follows:
\begin{align*}
	\PRC(X_{1}\Rightarrow Y\mid E=e) &= \operatorname{PostTCE}(X_{1}\Rightarrow Y \mid E=e), 
	\\
	\PRC(X_{2}\Rightarrow Y\mid E=e) &= \frac{\Pr(X_{2}=1\mid X_{1}=0)}{\Pr(X_{2}=1\mid X_{1}=1)} \operatorname{PostTCE}(X_{2}\Rightarrow Y \mid E=e). 
\end{align*}
More generally, for any variable that is a root node in the causal diagram, PRC coincides with PostTCE.  For non-root variables, however, PostTCE tends to overestimate their importance when caring about the root cause, which may yield rankings inconsistent with what root cause attribution requires.

\section{Numerical Examples}\label{sec:numerical-examples}

\subsection{Risk factors for hypertension}
This section employs an example of hypertension from \cite{tan2016Introduction} to illustrate the proposed approach. 
Figure \ref{fig:blood-pressure} illustrates the causal network and conditional probabilities, which model the relationship between a patient's hypertension and some risk factors, replicating \cite[Figure 1]{lu2023Evaluating}. 
Here, $X_1=1$ denotes no daily exercise, $X_2=1$ denotes an unhealthy diet, $X_3=1$ denotes the occurrence of heartburn, $X_4=1$ denotes the occurrence of heart disease, $X_5=1$ denotes the occurrence of chest pain, and $Y=1$ denotes the occurrence of high blood pressure. 
The variables are assumed to satisfy Assumptions \ref{ass:mono} and \ref{ass:no-conf}. 
While we use a fully specified Bayesian network to provide the ground truth in this illustrative example, we emphasize that PRC computation requires only the joint distribution and the topological order.

\begin{figure}[!ht]
	\centering
	\includegraphics[width = 0.6\linewidth]{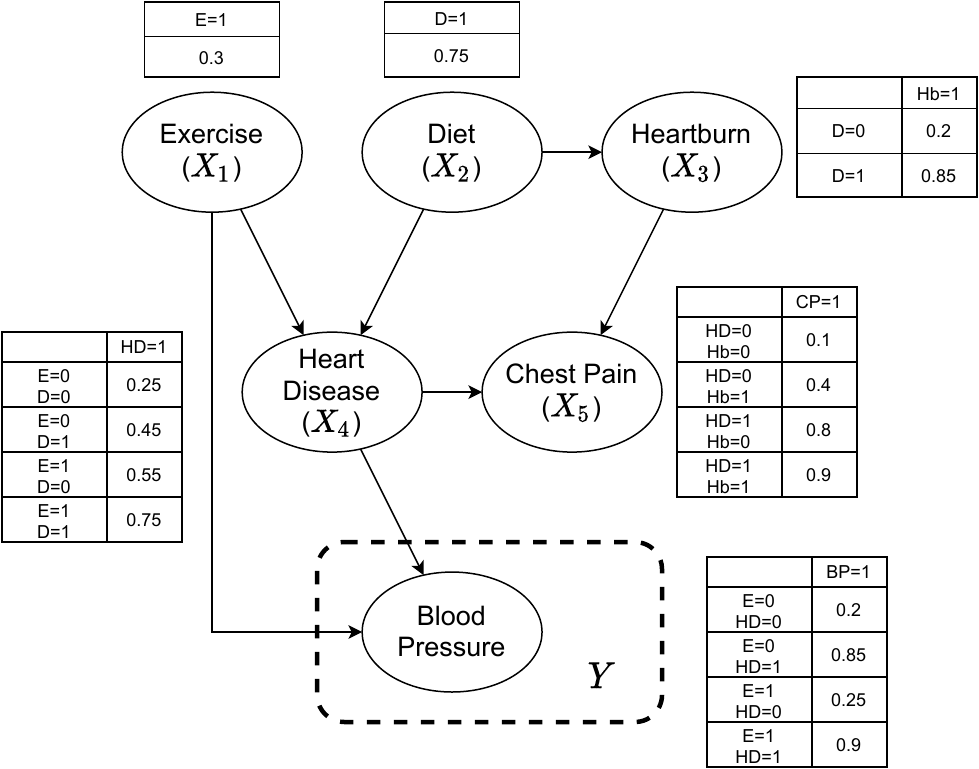}
	\caption{A causal network representing hypertension and its risk factors.\label{fig:blood-pressure}}
\end{figure}

\begin{table}[!ht]
	\centering
	\caption{Results of PRC and PostTCE under the evidence $E=(X,Y)=(1,1,1,1,1,1)$. \label{tab:ex1-1}}
	\begin{tabular}{cccccccc}
		\hline
											 & $X_1$  & $X_2$  & $X_3$  & $X_4$  & $X_5$  & $(X_{1},X_{2})$ & $\emptyset$\\ \hline
		$\PRC(X_{\Kc}\Rightarrow Y\mid E=e)$                 & 0.3444 & 0.1926 & 0      & 0.2407 & 0    & 0.5370 & 0.2222  \\
		$PostTCE(X_{\Kc}\Rightarrow Y\mid E=e)$              & 0.3444 & 0.1926 & 0      & 0.7222 & 0    & 0.5370 & \# \\
		\hline 
	\end{tabular}
\end{table}

Initially, consider the scenario where all the values of variables are observed as $x=(1,1,1,1,1)$ and $y=1$. 
See Table \ref{tab:ex1-1} for the results. 
All these causal-based methods rule out $X_3$ and $X_5$ as the cause, adhering to the principle of causality. 
Notably, PRC identifies $(X_{1},X_{2})$ as the root cause of $Y$ with the highest probability of $0.5370$, placing high emphasis on the root nodes of the causal network compared to PostTCE. 
PRC also suggests a probability of $0.2222$ that no root cause exists among $(X_{1},\dots,X_{5})$. 
PostTCE, by contrast, selects $X_4$ as the primary cause of $Y$ with a value of $0.7222$. 
Even though PostTCE coincides with PRC for root nodes, it assigns a higher value to $X_4$ than to $X_2$, due to the proximity of $X_4$ to $Y$ on the causal network. 
This tendency may lead to an overemphasis on proximate causes, while PRC can integrate the entire causal pathway, providing a more holistic view of the root causes of $Y$.

\begin{table}[!ht]
	\centering
	\caption{Results of PRC, PostTCE, and posterior probability under the evidence $Y=1$}
	\label{tab:ex1-2}
	\begin{tabular}{cccccccc}
		\hline
		& $X_1$  & $X_2$  & $X_3$  & $X_4$  & $X_5$ & $(X_{1},X_{2})$ & $\emptyset$ \\\hline
		$\PRC(X_{\Kc}\Rightarrow Y\mid E=e)$            & 0.1378 & 0.1828 & 0      & 0.3046 & 0   & 0.3205 & 0.3749   \\
		$PostTCE(X_{\Kc}\Rightarrow Y\mid E=e)$         & 0.1378 & 0.1828 & 0      & 0.5970 & 0    & 0.3205 & \#  \\
		$\Pr(X_{\Kc}=1\mid E=e)$                        & 0.3964 & 0.7957 & 0.7172 & 0.8004 & 0.7575 & 0.8811 & \#   \\\hline
	\end{tabular}
\end{table}

Consider a medical diagnosis scenario where only the symptom $Y$ is observed as $y=1$. 
See Table \ref{tab:ex1-2} for results. 
The posterior probability of $X_3$ and $X_5$ are $0.7172$ and $0.7575$ respectively, both exceeding $\Pr(X_{1}=1\mid Y=1)=0.3964$, despite their lack of a causal relationship with $Y$. 
The PostTCE of $X_4$ is $0.5970$, which is significantly higher than that of $X_2$ because $X_4$ blocks the pathway from $X_2$ to $Y$. 
The ranking of PRC for each variable is consistent with those of PostTCE: $X_4,X_2,X_1,X_3,X_5$.
However, PRC assigns a higher probability to the variable set $(X_1,X_2)$ than $X_{4}$, reflecting the influence of upstream variables. 
Moreover, PRC suggests a probability of $0.3749$ that high blood pressure ($Y=1$) is self-caused or unlisted residual factors.

Comparing the two evidence scenarios illustrates that root causes can change with the observed evidence. 
The structure of the causal network alone is not sufficient to identify the root cause. 

\subsection{Risk factors for harmonic gear drive device }
In this section, we use an example from \cite{yang2021Reliability}. 
Figure \ref{fig:HGD} illustrates the causal network and conditional probabilities, which are identical to \cite[Fig. 14]{yang2021Reliability} except for some modified conditional probabilities for presentational purposes.
Figure \ref{fig:HGD} shows the system configuration of a harmonic gear drive device, which comprises two main parts, i.e., the wave generation subsystem and the gear transmission subsystem, and can be further decomposed into five components. 
We code each variable as $1$ for failure and $0$ otherwise. To illustrate the proposed approach, we assume the variables satisfy the assumptions \ref{ass:mono} and \ref{ass:no-conf}.

\begin{figure}[!ht]
	\centering
	\includegraphics[width = 0.65\linewidth]{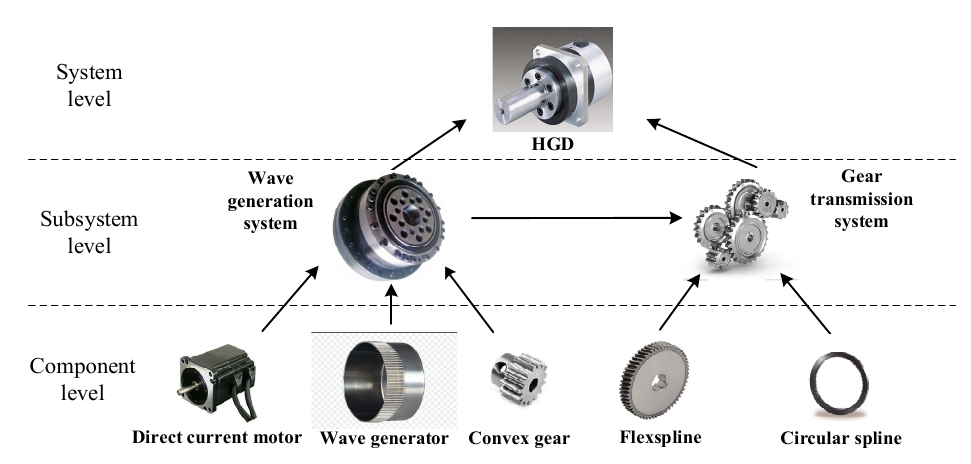}
	\includegraphics[width = 0.65\linewidth]{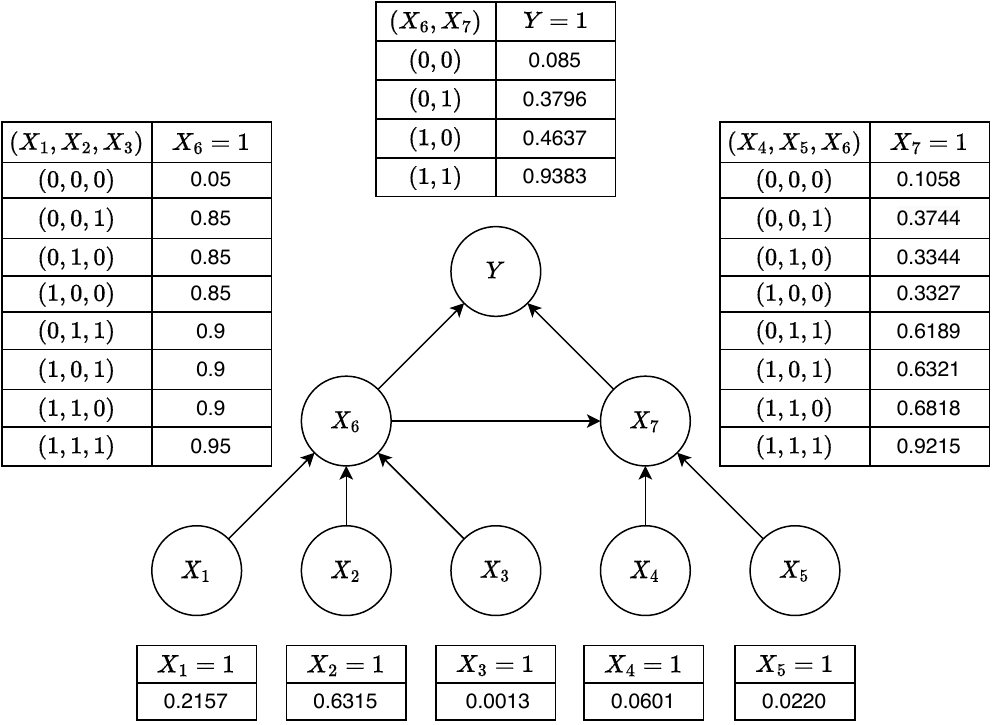}
	\caption{A causal network for the harmonic gear drive device}
	\label{fig:HGD}
\end{figure}

\begin{table}[!ht]
	\centering
	\caption{Results of PRC, PostTCE, and posterior probability under the evidence $Y=1$} 
	\label{tab:ex2-1}
	\begin{tabular}{cccccc}
		\hline
		& $X_1$  & $X_2$  & $X_3$  & $X_4$  & $X_5$  \\\hline
		$\PRC(X_{\Kc}\Rightarrow Y\mid E=e)$    & 0.0803 & 0.4595 & 0.0004 & 0.0045 & 0.0015 \\
		$PostTCE(X_{\Kc}\Rightarrow Y\mid E=e)$ & 0.0803 & 0.4595 & 0.0004 & 0.0045 & 0.0015 \\
		$\Pr(\lor_{k\in \Kc} X_k=1\mid E=e)$                & 0.2787 & 0.8008 & 0.0017 & 0.0644 & 0.0235 \\\hline 
		& $X_6$  & $X_7$  & $(X_1,X_{2},X_3)$  & $(X_{4},X_{5})$  & $\emptyset$  \\\hline
		$\PRC(X_{\Kc}\Rightarrow Y\mid E=e)$    & 0.0571 & 0.1032 & 0.6574 & 0.0060 & 0.2035 \\
		$PostTCE(X_{\Kc}\Rightarrow Y\mid E=e)$ & 0.7145 & 0.2290 & 0.6574 & 0.0060 & \#  \\
		$\Pr(\lor_{k\in \Kc} X_k=1\mid E=e)$                & 0.8932 & 0.4317 & 0.9011 & 0.0863 & \# \\ \hline
	\end{tabular}
\end{table}

Table \ref{tab:ex2-1} presents the results of PRC, PostTCE, and posterior probability under evidence $Y=1$. 
For the root nodes $X_{1},\dots,X_{5}$, the values of PRC are consistent with those of PostTCE. 
For the non-root variables $X_{6}$ and $X_{7}$, however, the two measures diverge substantially.  PostTCE assigns $X_{6}$ a value of $0.7145$, far exceeding its PRC of $0.0571$. 
Because of the strong causal dependence of $X_{6}$ on $X_{2}$, the failure of $X_{6}$ is largely attributable to $X_{2}$ rather than to $X_{6}$ itself, and PRC correctly reflects this by ranking $X_{2}$ (and the set $(X_{1},X_{2},X_{3})$) well above $X_{6}$, whereas PostTCE assigns $X_{6}$ a much higher value than $X_{2}$, mirroring the Property~\ref{ppt:TCE}.

\begin{table}[!ht]
	\centering
	\caption{Results of PRC, PostTCE, and posterior probability given the evidence $(X_1, X_2, Y)=(0,0,1)$} 
	\label{tab:ex2-2}
	\begin{tabular}{cccccc}
		\hline
		& $X_1$  & $X_2$  & $X_3$  & $X_4$  & $X_5$  \\\hline
		$\PRC(X_{\Kc}\Rightarrow Y\mid E=e)$    & 0     & 0     & 0.0035 & 0.0155 & 0.0057 \\
		$PostTCE(X_{\Kc}\Rightarrow Y\mid E=e)$ & 0     & 0     & 0.0035  & 0.0155 & 0.0057  \\
		$\Pr(\lor_{k\in \Kc} X_k=1\mid E=e)$    & 0     & 0     & 0.0047  & 0.0747 & 0.0276 \\\hline
		& $X_6$  & $X_7$  & $(X_1,X_{2},X_3)$  & $(X_{4},X_{5})$  & $\emptyset$  \\\hline
		$\PRC(X_{\Kc}\Rightarrow Y\mid E=e)$    & 0.1660 & 0.2238 & 0.0035 & 0.0212 & 0.5919 \\
		$PostTCE(X_{\Kc}\Rightarrow Y\mid E=e)$ & 0.1695 & 0.2734 & 0.0035 & 0.0212 & \#     \\
		$\Pr(\lor_{k\in \Kc} X_k=1\mid E=e)$    & 0.2119 & 0.3847 & 0.0047 & 0.1002 & \#     \\ \hline
	\end{tabular}
\end{table}

Consider a scenario in which we observe $(X_1, X_2)=(0,0)$ and $Y=1$. Table~\ref{tab:ex2-2} reports the corresponding results.  
In this case, it is not wise to select root nodes as root causes based solely on network structure, as $X_{3},X_{4},X_{5}$ are rare. 
The values of PRC closely align with those of PostTCE and posterior probability, yielding the same ranking $(X_{7},X_{6},X_{4},X_{5},X_{3},X_{1},X_{2})$ for the root cause. 
However, PRC assigns a probability of $0.5919$ to the event that none of $X_{1},\dots,X_{7}$ is the root cause of the system failure, suggesting that the root cause is $Y$ itself.

\section{Discussion}\label{sec:discussion}

We have proposed a causal framework for root cause attribution that addresses a conceptual gap in both the root cause analysis literature, which targets the right question but lacks a causal foundation, and the causal attribution literature, which provides rigorous causal answers but to a broader question. 
At its core is a formal, individual-level 
definition of a root cause within the potential outcomes framework, combining the notion of an individual cause with a {root condition} motivated by mediation analysis. 
The induced measure, the probability of root cause, is identifiable under standard assumptions, and admits an explicit identification formula.

Several directions remain open for future investigation. 
Although the identification formula of PRC was developed for binary variables, the underlying notion of root causes is considerably more general and can extend naturally to continuous variables and potentially non-Euclidean variables. 
Extending the identification theory and estimation methodology to these settings is a substantive and practically important challenge. 
The assumptions required for point identification, including monotonicity, are strong and, in general, untestable from observed data alone, which may limit the applicability of PRC in practice. 
Establishing sharp partial identification results and bounding strategies under weaker conditions is therefore an important direction for future work; related approaches in the literature on probabilities of causation \citep{joffe2001Using, dawid2022Bounding, zhang2022Partial} provide a natural starting point. 
Finally, the development of efficient semiparametric or nonparametric estimators for PRC, together with rigorous uncertainty quantification, represents another avenue deserving further attention \citep{cuellar2022Causes,tian2025Semiparametric}.

\bibliography{ref-PRC} 

\newpage 
\appendix

\section{Proof of Property~\ref{ppt:TCE}}\label{proof:ppt_TCE}
\begin{proof}
	\begin{align*}
		&\quad PostTCE(X_k\Rightarrow Y\mid E=e)
		\\
		&= \EE\bigl(Y_{X_k=1}-Y_{X_k=0}\mid E=e\bigr) 
		\\
		&= \sum_{x_i^*=0}^{1}\Pr\bigl\{Y_{X_k=1}=1,(X_i)_{X_k=1}=x_i^*\mid E=e\bigr\} -  \sum_{x_i^*=0}^{1}\Pr\bigl\{Y_{X_k=0}=1,(X_i)_{X_k=0}=x_i^*\mid E=e\bigr\} 
		\\
		&= \sum_{x_i^*=0}^{1}\Pr\bigl\{Y_{X_k=1,X_i=x_i^*}=1,(X_i)_{X_k=1}=x_i^*\mid E=e\bigr\} -  \sum_{x_i^*=0}^{1}\Pr\bigl\{Y_{X_k=0,X_i=x_i^*}=1,(X_i)_{X_k=0}=x_i^*\mid E=e\bigr\} 
		. 
	\end{align*}
		Because of $X_i$ blocking all the causal path from $X_k$ to $Y$ and the assumption of no confounding and cross-world independence, for $a=0,1$, we have 
	\begin{align*}
		&\quad \Pr\bigl\{Y_{X_k=a,X_i=x_i^*}=1,(X_i)_{X_k=a}=x_i^*\mid E=e\bigr\}
		\\
		&= \Pr\bigl\{Y_{X_i=x_i^*}=1,(X_i)_{X_k=a}=x_i^*\mid E=e\bigr\}
		\\
		&= \Pr\bigl(Y_{X_i=x_i^*}=1\mid E=e\bigr) \Pr\bigl\{(X_i)_{X_k=a}=x_i^*\mid E=e\bigr\}
	\end{align*}
	Thus, 
	\begin{align*} 
		&\quad PostTCE(X_k\Rightarrow Y\mid E=e)
		\\
		&= \sum_{x_i^*=0}^{1}\Pr\bigl(Y_{X_i=x_i^*}=1\mid E=e\bigr) \Pr\bigl\{(X_i)_{X_k=1}=x_i^*\mid E=e\bigr\} -  \sum_{x_i^*=0}^{1}\Pr\bigl(Y_{X_i=x_i^*}=1\mid E=e\bigr) \Pr\bigl\{(X_i)_{X_k=0}=x_i^*\mid E=e \bigr\}
		\\
		&= \sum_{x_i^*=0}^{1}\Pr\bigl(Y_{X_i=x_i^*}=1\mid E=e\bigr) \Bigl[\Pr\bigl\{(X_i)_{X_k=1}=x_i^*\mid E=e\bigr\} - \Pr\bigl\{(X_i)_{X_k=0}=x_i^*\mid E=e\bigr\} \Bigr]
		\\
		&= \Bigl\{\Pr\bigl(Y_{X_i=1}=1\mid E=e\bigr) - \Pr\bigl(Y_{X_i=1}=0\mid E=e\bigr)\Bigr\} \cdot \Bigl[\Pr\bigl\{(X_i)_{X_k=1}=1\mid E=e\bigr\} - \Pr\bigl\{(X_i)_{X_k=0}=1\mid E=e\bigr\}\Bigr]
		\\
		&= PostTCE(X_i\Rightarrow Y\mid E=e) \cdot \EE\bigl\{(X_i)_{X_k=1} - (X_i)_{X_k=0}\mid E=e\bigr\} 
		. 
	\end{align*}
	Because $-1\leq \EE\bigl\{(X_i)_{X_k=1} - (X_i)_{X_k=0}\mid E=e\bigr\} \leq 1$, we have 
	\begin{align*} 
		\bigl| PostTCE(X_k\Rightarrow Y\mid E=e) \bigr| \leq \bigl| PostTCE(X_i\Rightarrow Y\mid E=e) \bigr|
		. 
	\end{align*}
\end{proof}

\section{Proof of Property~\ref{ppt:PRC}}\label{proof:ppt_PRC}

\begin{proof}\ 
	\begin{enumerate}[(a)]
		\item If $\Cc_{X_{K}}^{Y}=1$, the inequality holds. 
			If $\Cc_{X_{K}}^{Y}=0$, then for all $x_{K}'\neq x_{K}''$, we have $Y_{x_{K}'}=Y_{x_{K}''}$, therefore $\Rc_{X_{K}}^{Y}=1$. 

			We give an couterexample here to show that the root condition does not satisfy the monotonicity. 
			\begin{example}\label{ex:RC-non-monotonicity}
				The data generating process of binary variables $(X_{1},X_{2},X_{3},Y)$ is as follows: 
				\begin{align*}\left\{\begin{aligned}
					X_{1} &= f_{1}(\epsilon_{1}), \\
					X_{2} &= f_{2}(X_{1})=X_{1}, \\
					X_{3} &= f_{3}(\epsilon_{3}), \\
					Y &= f(X) = X_{2}X_{3} + (1-X_{2})(1-X_{3}),
				\end{aligned}\right.\end{align*}
				where $(\epsilon_{1},\epsilon_{3})$ are independent random variables. 
				Then, we can simply verified the following results:
				\begin{enumerate}[(a)]
					\item $X_{3}$ is the RC of $Y$. \par 
						$Y_{X_{3}=1}=X_{2}$, $Y_{X_{3}=0}=1-X_{2}$, and $\Cc_{X_{3}}^{Y}= \lor\{Y_{X_{3}=1}\neq Y_{X_{3}=0}\}=1$, then $X_{3}$ is the IC of $Y$. $\Rc_{X_{3}}^{Y} = \land_{x_{\{1,2\}}'\neq x_{\{1,2\}}''} \{Y_{(X_{3})_{x_{\{1,2\}}'}}=Y_{(X_{3})_{x_{\{1,2\}}''}}\} = \land\{Y_{X_{3}}=Y_{X_{3}}\}=1$, then $X_{3}$ satisfies the root condition.
					\item $X_{\{2,3\}}$ is not the RC of $Y$. \par 
						$\Cc_{X_{\{2,3\}}}^{Y}\geq \Cc_{X_{3}}=1$, then $X_{2,3}$ is the IC of $Y$. $Y_{(X_{\{2,3\}})_{X_{1}=1}}=X_{3}$, $Y_{(X_{\{2,3\}})_{X_{1}=0}}=1-X_{3}$, $\Rc_{X_{\{2,3\}}}^{Y} = \land \{Y_{(X_{\{2,3\}})_{X_{1}=1}}=Y_{(X_{\{2,3\}})_{X_{1}=0}}\}=0$, then $X_{2,3}$ does not satisfy the root condition. 
				\end{enumerate} 
			\end{example}
		\item 
			\begin{align*}
				\Cc_{X_{K}}^{Y} 
				= \bigvee_{x_{K}'\neq x_{K}''} \{Y_{x_{K}'}\neq Y_{x_{K}''}\}
				\geq \max_{k\in K} \bigvee_{x_{k}'\neq x_{k}''} \bigl\{Y_{x_{k}',(X_{K\setminus \{k\}})_{x_{k}'}}\neq Y_{x_{k}'',(X_{K\setminus \{k\}})_{x_{k}''}}\bigr\}
				= \bigvee_{k\in K} \Cc_{X_{k}}^{Y}.
			\end{align*}
		\item $\bigl\{\Cc_{X_{k}}^{Y}= 1\bigr\}  = \{Y_{X_{k}=0}\neq Y_{X_{k}=1}\}$ and 
			\begin{align*}
				\bigl\{\Cc_{X_{k}}^{Y} = 1, \Rc_{X_{k}}^{Y}=1\bigr\} 
				&= \{Y_{X_{k}=0}\neq Y_{X_{k}=1}\} \land \bigl\{ \forall\ \ol x_{k}'\neq \ol x_{k}'', Y_{(X_{k})_{\ol x_{k}'}}= Y_{(X_{k})_{\ol x_{k}''}}\bigr\}
				\\
				&= \{Y_{X_{k}=0}\neq Y_{X_{k}=1}\} \land \bigl\{\forall\ \ol x_{k}'\neq \ol x_{k}'', (X_{k})_{\ol x_{k}'}=(X_{k})_{\ol x_{k}''}\bigr\}
				\\
				&= \{Y_{X_{k}=0}\neq Y_{X_{k}=1}\} \land \bigl\{L_{k}=1\bigr\}
				\\
				&= \{Y_{X_{k}=0}\neq Y_{X_{k}=1}, L_{k}=1\}. 
			\end{align*}
		\item $K=\{1,\dots,k\}$ then $\ol X_{K}=\emptyset$ and $\Rc_{X_{K}}^{Y}=1$. 
			\begin{align*}
				\PRC(X_{K}\Rightarrow Y\mid E=e) 
				&= \Pr\bigl(\Cc_{X_{K}}^{Y}=1\mid E=e\bigr)
				\\
				&\geq \max_{\Jc\subseteq K} \Pr\bigl(\Cc_{X_{\Jc}}^{Y}=1\mid E=e\bigr)
				\\
				&\geq \max_{\Jc\subseteq K} \Pr\bigl(\Cc_{X_{\Jc}}^{Y}=1, \Rc_{X_{\Jc}}^{Y}=1\mid E=e\bigr)
				\\
				&= \max_{\Jc\subseteq K} \PRC(X_{\Jc}\Rightarrow Y\mid E=e). 
			\end{align*}
	\end{enumerate}
	\end{proof}

\section{Proof of Theorem~\ref{thm:PRC}}\label{proof:thm_PRC}
\begin{proof}
	For convenience, with a index set $K=\{k_{1},\dots, k_{r}\}$ with $1\leq k_{1}<k_{2}< \dots < k_{r}\leq p$, we denote $X_{K}=(X_{k_{1}},\dots, X_{k_{r}})$, $\ol K = \{j\mid 1\leq j<k_{1}\}$, $\ul K = \{j\mid k_{r}<j\leq p\}$ and $\widehat K=\{j\mid k_{1}<j<k_{r},j\notin I\}$. 

	Because of Assumption~\ref{ass:mono} of monotonicity, we have $Y_{X_{K}=0}\leq Y_{x_{K}'} \leq Y_{X_{K}=1}$, $(X_{K})_{X_{\ol K\cup \hat K}=0}\preceq (X_{K})_{x_{\ol K\cup \hat K}'} \preceq (X_{K})_{X_{\ol K\cup \hat K}=1}$, and $Y_{(X_{K})_{X_{\ol K\cup \hat K}=0}} \leq Y_{(X_{K})_{x_{\ol K\cup \hat K}'}} \leq Y_{(X_{K})_{X_{\ol K\cup \hat K}=1}}$. Therefore, 
	\begin{align*}
		\Cc_{X_{K}}^{Y} 
		= \bigvee_{x_{K}'\neq x_{K}''} \bigl\{Y_{x_{K}'}\neq Y_{x_{K}''}\bigr\}
		= \bigl\{Y_{X_{K}=0}\neq Y_{X_{K}=1}\bigr\} 
		= \bigl\{Y_{X_{K}=0}=0, Y_{X_{K}=1}=1\bigr\} 
		, 
	\end{align*}
	\begin{align*}
		\Rc_{X_{K}}^{Y}
		&= \bigwedge_{x_{K}'\neq x_{K}''} \bigl\{Y_{(X_{K})_{x_{\ol K\cup \hat K}'}}=Y_{(X_{K})_{x_{\ol K\cup \hat K}''}}\bigr\} 
		= \bigl\{Y_{(X_{K})_{X_{\ol K\cup \hat K}=0}}=Y_{(X_{K})_{X_{\ol K\cup \hat K}=1}}\bigr\} 
		\\
		&= \bigl\{Y_{(X_{K})_{X_{\ol K\cup \hat K}=0}}=Y_{(X_{K})_{X_{\ol K\cup \hat K}=1}}=Y_{(X_{K})_{X_{\ol K\cup \hat K}=x_{\ol K\cup \hat K}}}\bigr\} 
		. 
	\end{align*}
	Thus, 
	\begin{align*}
		&\quad \PRC(X_{K}\Rightarrow Y\mid X=x, Y=1) 
		\\
		&= \Pr\bigl(\Cc_{X_{K}}^{Y}=1,\Rc_{X_{K}}^{Y}=1\mid X=x,Y=1\bigr)
		\\
		&= \frac{1}{\Pr(X=x,Y=1)} \Pr\bigl\{Y_{X_{K}=0}\neq Y_{X_{K}=1}, Y_{(X_{K})_{X_{\ol K\cup \hat K}=0}}=Y_{(X_{K})_{X_{\ol K\cup \hat K}=1}}, X=x, Y=1\bigr\}
		\\
		&= \frac{1}{\Pr(X=x,Y=1)} \Pr\bigl\{Y_{X_{K}=0}=0, Y_{X_{K}=1}=1, Y_{(X_{K})_{X_{\ol K\cup \hat K}=0}}=Y_{(X_{K})_{X_{\ol K\cup \hat K}=1}}=Y_{(X_{K})_{x_{\ol K\cup \hat K}}}=Y=1, X=x\bigr\} 
		\\
		&= \frac{1}{\Pr(X=x,Y=1)} \Pr\bigl\{Y_{X_{K}=0}=0, Y_{(X_{K})_{X_{\ol K\cup \hat K}=0}}=1, X=x\bigr\}
		, 
	\end{align*}
	where the last equation holds because of Assumption~\ref{ass:mono} of monotonicity. Then we have 
	\begin{align*}
		\Pr(X=x,Y=1) &= \Pr(Y=1\mid X=x) \prod_{i=1}^{p} \Pr(X_{i}=x_{i}\mid \ol X_{i}=\ol x_{i}), 
	\end{align*}
	\begin{align*}
		&\quad \Pr\bigl\{Y_{X_{K}=0}=0, Y_{(X_{K})_{X_{\ol K\cup \hat K}=0}}=1, X=x\bigr\}
		\\
		&= \sum_{x_{K}'\preceq x_{K}}\Pr\bigl\{Y_{X_{K}=0}=0, Y_{(X_{K})_{X_{\ol K\cup \hat K}=0}}=1, X=x, (X_{K})_{X_{\ol K\cup \hat K}=0}=x_{K}'\bigr\}
		\\
		&= \sum_{x_{K}'\preceq x_{K}}\Pr\bigl\{Y_{X_{K}=0}=0, Y_{x_{K}'}=1, X=x, (X_{K})_{X_{\ol K\cup \hat K}=0}=x_{K}'\bigr\}
		\\
		&= \sum_{x_{\hat K \cup \ul K}^{*}\preceq x_{\hat K \cup \ul K}} \sum_{x_{K\cup \hat K \cup \ul K}'\preceq x_{K\cup \hat K \cup \ul K}} \Pr\Bigl\{Y_{x_{\ol K}, X_{K}=0, x_{\hat K \cup \ul K}^{*}}=0, Y_{x_{\ol K},x_{K}',x_{\hat K \cup \ul K}'}=1, X=x, 
		\\&\qquad\qquad\qquad\qquad\qquad\qquad\qquad\qquad (X_{\hat K \cup \ul K})_{X_{K}=0} = x_{\hat K \cup \ul K}^{*}, (X_{K})_{X_{\ol K\cup \hat K}=0}=x_{K}', (X_{\hat K\cup \ul K})_{x_{K}'}=x_{\hat K\cup \ul K}' \Bigr\}
		\\
		&= \sum_{\substack{x^*\preceq x' \preceq x: \\ x_{\ol K}^{*}=x_{\ol K}'=x_{\ol K},\ x_{K}^{*}=0}} \Pr\Bigl\{Y_{x^{*}}=0, Y_{x'}=1, X=x, (X_{\hat K \cup \ul K})_{X_{K}=0} = x_{\hat K \cup \ul K}^{*}, (X_{K})_{X_{\ol K\cup \hat K}=0}=x_{K}', (X_{\hat K\cup \ul K})_{x_{K}'}=x_{\hat K\cup \ul K}' \Bigr\}
		. 
	\end{align*}
	For convenience, we denote $x_{\ol K}^{*}=x_{\ol K}'=x_{\ol K}$ and $x_{K}^{*}=0$. Based on the monotonicity assumption, we have $x^*\preceq x' \preceq x$. 
	Let  and $\ul K = \{i\mid \min(K)<i \leq p, i\notin K\}$. 
	
	\begin{align*}
		&\quad \Pr\Bigl\{Y_{x^{*}}=0, Y_{x'}=1, X=x, (X_{\hat K \cup \ul K})_{X_{K}=0} = x_{\hat K \cup \ul K}^{*}, (X_{K})_{X_{\ol K\cup \hat K}=0}=x_{K}', (X_{\hat K\cup \ul K})_{x_{K}'}=x_{\hat K\cup \ul K}' \Bigr\}
		\\
		&= \Pr\bigl\{Y_{x^{*}}=0, Y_{x'}=1\bigr\} \times \prod_{i\in \ol K} \Pr(X_{i}=x_{i}\mid \ol X_{i}=\ol x_{i}) \times\prod_{i\in K} \Pr\bigl\{(X_{i})_{\ol x_{i}}=x_{i}, (X_{i})_{\ol x_{i}''}=x_{i}' \bigr\} 
		\\&\qquad \times \prod_{i\in \hat K\cup \ul K} \Pr\bigl\{(X_{i})_{\ol x_{i}}=x_{i}, (X_{i})_{\ol x_{i}'}=x_{i}', (X_{i})_{\ol x_{i}^{*}}=x_{i}^{*}\bigr\}
		, 
	\end{align*}
	where $x_{K}''=x_{K}'$, $x_{\ol K\cup \hat K}''=0$, and $x_{K}''\preceq x_{K}$. 
	Consider each term in the product. 
	\begin{align*}
		&\quad \Pr(Y_{x^{*}}=0, Y_{x'}=1) 
		\\
		&= \Pr(Y_{x'}=1) - \Pr(Y_{x^{*}}=1, Y_{x'}=1) 
		\\
		&= \Pr(Y_{x'}=1) - \Pr(Y_{x^{*}}=1) 
		\\
		&= \Pr(Y=1 \mid x') - \Pr(Y=1 \mid x^{*})
		, 
	\end{align*}
	where the second equation holds because of the monotonicity assumption, and the last equation holds because of the no-confounding assumption. 

	For $i\in K$, we have
	\begin{align*}
		&\quad \Pr\bigl\{(X_{i})_{\ol x_{i}}=x_{i}, (X_{i})_{\ol x_{i}''}=x_{i}'' \bigr\} 
		\\
		&= x_{i}'' x_{i} \Pr\bigl\{(X_{i})_{\ol x_{i}}=1, (X_{i})_{\ol x_{i}''}=1 \bigr\} + (1-x_{i}'')(1-x_{i}) \Pr\bigl\{(X_{i})_{\ol x_{i}}=0, (X_{i})_{\ol x_{i}''}=0 \bigr\} 
		\\&\qquad + (1-x_{i}'')x_{i} \Pr\bigl\{(X_{i})_{\ol x_{i}}=1, (X_{i})_{\ol x_{i}''}=0 \bigr\} +  x_{i}'' (1-x_{i}) \Pr\bigl\{(X_{i})_{\ol x_{i}}=0, (X_{i})_{\ol x_{i}''}=1 \bigr\}
		\\
		&= x_{i}'' x_{i} \Pr\bigl\{(X_{i})_{\ol x_{i}''}=1 \bigr\} + (1-x_{i}'')(1-x_{i}) \Pr\bigl\{(X_{i})_{\ol x_{i}}=0\bigr\} 
		\\&\qquad + (1-x_{i}'')x_{i} \Bigl[\Pr\bigl\{(X_{i})_{\ol x_{i}}=1 \bigr\} - \Pr\bigl\{(X_{i})_{\ol x_{i}''}=1 \bigr\}\Bigr] +  x_{i}'' (1-x_{i})\cdot 0
		\\
		&= (2x_{i}''-1) x_{i} \Pr\bigl\{(X_{i})_{\ol x_{i}''}=1 \bigr\} + (1-x_{i}'')\Bigl[(1-x_{i}) \Pr\bigl\{(X_{i})_{\ol x_{i}}=0\bigr\} + x_{i} \Pr\bigl\{(X_{i})_{\ol x_{i}}=1 \bigr\} \Bigr]
		\\
		&= (2x_{i}''-1) x_{i} \cdot \Pr\bigl\{(X_{i})_{\ol x_{i}''}=1 \bigr\} + (1-x_{i}'') \cdot \Pr\bigl\{(X_{i})_{\ol x_{i}}=x_{i}\bigr\}
		\\
		&= (2x_{i}''-1) x_{i} \cdot \Pr(X_{i}=1 \mid \ol X_{i}=\ol x_{i}'') + (1-x_{i}'') \cdot \Pr(X_{i}=x_{i}\mid \ol X_{i}=\ol x_{i})
		\\
		&= (2x_{i}'-1) x_{i} \cdot \Pr(X_{i}=1 \mid \ol X_{i}=\ol x_{i}'') + (1-x_{i}') \cdot \Pr(X_{i}=x_{i}\mid \ol X_{i}=\ol x_{i}), 
	\end{align*}
	where $x''_{\ol K\cup\hat K}=0$, $x_{K}''=x_{K}'$; the second equation holds because of the monotonicity assumption, and the penultimate equation holds because of the no-confounding assumption. 

	For $i\in \hat K\cup \ul K$, we have
	\begin{align*}
		&\quad \Pr\bigl\{(X_{i})_{\ol x_{i}}=x_{i}, (X_{i})_{\ol x_{i}'}=x_{i}', (X_{i})_{\ol x_{i}^*}=x_{i}^{*}\bigr\} 
		\\
		&= x_{i} \Pr\bigl\{(X_{i})_{\ol x_{i}}=1, (X_{i})_{\ol x_{i}'}=x_{i}', (X_{i})_{\ol x_{i}^*}=x_{i}^{*}\bigr\} + (1-x_{i}) \Pr\bigl\{(X_{i})_{\ol x_{i}}=0, (X_{i})_{\ol x_{i}'}=x_{i}', (X_{i})_{\ol x_{i}^*}=x_{i}^{*}\bigr\} 
		\\
		&= x_{i} x_{i}'(2x_{i}^{*}-1)\Pr\bigl\{(X_{i})_{\ol x_{i}}=1, (X_{i})_{\ol x_{i}^*}=1\bigr\} + x_{i}(1-x_{i}^{*})\Pr\bigl\{(X_{i})_{\ol x_{i}}=1, (X_{i})_{\ol x_{i}'}=x_{i}'\bigr\} 
		\\&\qquad + (1-x_{i})(1-x_{i}')(1-x_{i}^*) \Pr\bigl\{(X_{i})_{\ol x_{i}}=0\bigr\} 
		\\
		&= (1-x_{i})(1-x_{i}')(1-x_{i}^*) \Pr\bigl\{(X_{i})_{\ol x_{i}}=0\bigr\} + x_{i} x_{i}'(2x_{i}^{*}-1)\Pr\bigl\{(X_{i})_{\ol x_{i}^*}=1\bigr\} 
		\\&\qquad + x_{i}x_{i}'(1-x_{i}^{*})\Pr\bigl\{(X_{i})_{\ol x_{i}'}=1\bigr\} + x_{i}(1-x_{i}')(1-x_{i}^{*})\bigl[\Pr\bigl\{(X_{i})_{\ol x_{i}}=1\bigr\} - \Pr\bigl\{(X_{i})_{\ol x_{i}'}=1\bigr\} \bigr]
		\\
		&= (1-x_{i}')(1-x_{i}^*) \Pr\bigl\{(X_{i})_{\ol x_{i}}=x_{i}\bigr\} + x_{i} x_{i}'(2x_{i}^{*}-1)\Pr\bigl\{(X_{i})_{\ol x_{i}^*}=1\bigr\} + x_{i}(2x_{i}'-1)(1-x_{i}^{*})\Pr\bigl\{(X_{i})_{\ol x_{i}'}=1\bigr\} 
		\\
		&= (1-x_{i}')(1-x_{i}^*) \Pr\bigl(X_{i}=x_{i} \mid \ol X_{i}=\ol x_{i}\bigr) + x_{i} x_{i}'(2x_{i}^{*}-1)\Pr\bigl(X_{i}=1\mid \ol X_{i}=\ol x_{i}^*\bigr) 
		\\&\qquad + x_{i}(2x_{i}'-1)(1-x_{i}^{*})\Pr\bigl(X_{i}=1\mid \ol X_{i}=\ol x_{i}'\bigr) 
		, 
	\end{align*}
	where the second and third equations hold because of the monotonicity assumption, and the last equation holds because of the no-confounding assumption. 

	Therefore,
	\begin{align*}
		\frac{\Pr(Y_{x^{*}}=0, Y_{x'}=1)}{\Pr(Y=1\mid X=x)} 
		= \frac{\Pr(Y=1 \mid x') - \Pr(Y=1 \mid x^{*})}{\Pr(Y=1\mid X=x)}
		; 
	\end{align*}
	For $i\in K$,
	\begin{align*}
		\frac{ \Pr\bigl\{(X_{i})_{\ol x_{i}}=x_{i}, (X_{i})_{\ol x_{i}''}=x_{i}''\bigr\}  }{\Pr(X_{i}=x_{i}\mid \ol X_{i}=\ol x_{i})} 
		= (1-x_{i}') + (2x_{i}'-1) x_{i} \cdot \frac{\Pr(X_{i}=1 \mid \ol X_{i}=\ol x_{i}'')}{\Pr(X_{i}=x_{i}\mid \ol X_{i}=\ol x_{i})} 
		; 
	\end{align*}
	For $i\in \hat K\cup \ul K$, 
	\begin{align*}
		&\quad \frac{\Pr\bigl\{(X_{i})_{\ol x_{i}}=x_{i}, (X_{i})_{\ol x_{i}'}=x_{i}', (X_{i})_{\ol x_{i}^*}=x_{i}^{*}\bigr\}}{\Pr(X_{i}=x_{i}\mid \ol X_{i}=\ol x_{i})} 
		\\ 
		&= (1-x_{i}')(1-x_{i}^*) + x_{i} x_{i}'(2x_{i}^{*}-1) \frac{\Pr\bigl(X_{i}=1\mid \ol X_{i}=\ol x_{i}^*\bigr)}{\Pr(X_{i}=x_{i}\mid \ol X_{i}=\ol x_{i})}
		\\&\qquad + x_{i}(2x_{i}'-1)(1-x_{i}^{*}) \frac{\Pr\bigl(X_{i}=1\mid \ol X_{i}=\ol x_{i}'\bigr)}{\Pr(X_{i}=x_{i}\mid \ol X_{i}=\ol x_{i})}
		. 
	\end{align*}

	Combining the above results, we have 
	\begin{align*}
		&\quad \PRC(X_{K}\Rightarrow Y\mid X=x, Y=1)
		\\
		&= \frac{1}{\Pr(X=x,Y=1)} \Pr\bigl\{Y_{X_{K}=0}=0, Y_{(X_{K})_{X_{\ol K\cup \hat K}=0}}=1, X=x\bigr\}
		\\
		&= \sum_{\substack{x^*\preceq x' \preceq x: \\ x_{\ol K}^{*}=x_{\ol K}'=x_{\ol K},\ x_{K}^{*}=0}} \frac{\Pr\Bigl\{Y_{x^{*}}=0, Y_{x'}=1, X=x, (X_{\hat K \cup \ul K})_{X_{K}=0} = x_{\hat K \cup \ul K}^{*}, (X_{K})_{X_{\ol K\cup \hat K}=0}=x_{K}', (X_{\hat K\cup \ul K})_{x_{K}'}=x_{\hat K\cup \ul K}' \Bigr\}}{\Pr(X=x,Y=1)}
		\\
		&= \sum_{\substack{x^*\preceq x' \preceq x: \\ x_{\ol K}^{*}=x_{\ol K}'=x_{\ol K},\ x_{K}^{*}=0}} \frac{\Pr(Y_{x^{*}}=0, Y_{x'}=1)}{\Pr(Y=1\mid X=x)} \times \prod_{i\in \ol K} \frac{\Pr(X_{i}=x_{i}\mid \ol X_{i}=\ol x_{i})}{\Pr(X_{i}=x_{i}\mid \ol X_{i}=\ol x_{i})} \times\prod_{i\in K} \frac{\Pr\bigl\{(X_{i})_{\ol x_{i}}=x_{i}, (X_{i})_{\ol x_{i}''}=x_{i}' \bigr\}}{\Pr(X_{i}=x_{i}\mid \ol X_{i}=\ol x_{i})} 
		\\&\qquad\qquad\qquad \times \prod_{i\in \hat K\cup \ul K} \frac{\Pr\bigl\{(X_{i})_{\ol x_{i}}=x_{i}, (X_{i})_{\ol x_{i}'}=x_{i}', (X_{i})_{\ol x_{i}^{*}}=x_{i}^{*}\bigr\}}{\Pr(X_{i}=x_{i}\mid \ol X_{i}=\ol x_{i})}
		\\
		&= \sum_{\substack{x^*\preceq x' \preceq x: \\ x_{\ol K}^{*}=x_{\ol K}'=x_{\ol K},\ x_{K}^{*}=0}} \frac{\Pr(Y=1 \mid x') - \Pr(Y=1 \mid x^{*})}{\Pr(Y=1\mid X=x)} \times \prod_{i\in K} \Bigl\{ (1-x_{i}') + x_{i}(2x_{i}'-1) \cdot \frac{\Pr(X_{i}=1 \mid \ol X_{i}=\ol x_{i}'')}{\Pr(X_{i}=x_{i}\mid \ol X_{i}=\ol x_{i})} \Bigr\}
		\\&\qquad\times \prod_{i\in \hat K\cup \ul K} \Bigl\{ (1-x_{i}')(1-x_{i}^*) + x_{i}(2x_{i}'-1)(1-x_{i}^{*}) \frac{\Pr\bigl(X_{i}=1\mid \ol X_{i}=\ol x_{i}'\bigr)}{\Pr(X_{i}=1\mid \ol X_{i}=\ol x_{i})} + x_{i} x_{i}'(2x_{i}^{*}-1) \frac{\Pr\bigl(X_{i}=1\mid \ol X_{i}=\ol x_{i}^*\bigr)}{\Pr(X_{i}=1\mid \ol X_{i}=\ol x_{i})} \Bigr\}
		, 
	\end{align*}
	where $x_{K}''=x_{K}'$, $x_{\ol K\cup \hat K}''=0$. 
\end{proof}






\end{document}